%% file: OptimalNAGT.tex
\documentclass[english]{article}

\makeatletter

\input{preamble}

\usepackage{color}
\usepackage{url}
\usepackage{enumitem}
\usepackage{float}
\usepackage{comment}
\usepackage{amsthm}
\usepackage{breqn}

\floatstyle{ruled}
\newfloat{algorithm}{tbp}{loa}
\providecommand{\algorithmname}{Algorithm}
\floatname{algorithm}{\protect\algorithmname}

\usepackage[top=1in, left=1in, right=1in, bottom=1in]{geometry}

\usepackage{scalerel,stackengine}

\makeatletter
\newcommand{\manuallabel}[2]{\def\@currentlabel{#2}\label{#1}}
\makeatother

\numberwithin{equation}{section}
\numberwithin{thm}{section}
\numberwithin{defn}{section}
\numberwithin{lem}{section}
\numberwithin{cor}{section}

\begin{document}


\title{Optimal Non-Adaptive Probabilistic Group \\ Testing in General Sparsity Regimes}




\author{Wei Heng Bay, Eric Price, and Jonathan Scarlett}

\date{}
\maketitle

\input{abstract.tex}

\long\def\symbolfootnote[#1]#2{\begingroup\def\thefootnote{\fnsymbol{footnote}}\footnote[#1]{#2}\endgroup}

\symbolfootnote[0]{ W.H.~Bay and J.~Scarlett are with the  Department of Computer Science and the Department of Mathematics, National University of Singapore  (e-mail: \url{bayweiheng@gmail.com, scarlett@comp.nus.edu.sg}). J.~Scarlett is also with the Institute of Data Science, National University of Singapore. E.~Price is with the Department of Computer Science, University of Texas at Austin (e-mail: \url{ecprice@cs.utexas.edu}).

E.~Price was supported in part by NSF Award CCF-1751040 (CAREER).  J.~Scarlett was supported by an NUS Early Career Research Award.}

\input{main_body.tex}


\appendix

\input{appendix.tex}

\section*{Acknowledgment}

We are very grateful to an anonymous reviewer for helpful suggestions and encouraging us to establish the precise constant factors in the analysis, which allowed us to significantly strengthen an earlier version of our main result containing an unspecified constant.

\bibliographystyle{imaiai}
\bibliography{JS_References,HYP-biblio}

\end{document}

%% file: preamble.tex
\usepackage[T1]{fontenc}
\usepackage[latin9]{inputenc}
\usepackage{amsthm}
\usepackage{amsmath}
\usepackage{amssymb}
\usepackage{graphicx}
\usepackage{dsfont}
\usepackage{cite}
\usepackage{amsmath}
\usepackage{array}
\usepackage{multirow}
\usepackage{caption}
\usepackage{color}

\usepackage{multicol}

\theoremstyle{plain}

\theoremstyle{plain}

\theoremstyle{plain}
\newtheorem{lem}{\protect\lemmaname}
\theoremstyle{plain}
\newtheorem{thm}{\protect\theoremname}
\theoremstyle{plain}
\newtheorem{cor}{\protect\corollaryname}  
\theoremstyle{definition}
\newtheorem{defn}{\protect\definitionname}
\theoremstyle{definition}

\theoremstyle{definition}

\makeatother
  
\usepackage{babel} 

\providecommand{\claimname}{Claim}
\providecommand{\lemmaname}{Lemma}
\providecommand{\propositionname}{Proposition}
\providecommand{\theoremname}{Theorem}
\providecommand{\corollaryname}{Corollary} 
\providecommand{\definitionname}{Definition}
\providecommand{\assumptionname}{Assumption}
\providecommand{\remarkname}{Remark}

\newcommand{\overbar}[1]{\mkern 1.25mu\overline{\mkern-1.25mu#1\mkern-0.25mu}\mkern 0.25mu}

\newcommand{\Shat}{\widehat{S}}

\newcommand{\kbar}{\overbar{k}}

\newcommand{\Xv}{\mathbf{X}}

\newcommand{\PP}{\mathbb{P}}

\newcommand{\sett}[1]{\left\lbrace#1\right\rbrace}

\newcommand{\bb}[1]{\textbf{#1}}
\newcommand{\pr}{\mathbb{P}}

\usepackage{amsmath}
\usepackage{amssymb}
\usepackage{fullpage}
\usepackage{comment}
\usepackage{bbm}
\usepackage{graphicx}
\usepackage{xcolor}
\usepackage{amsthm}
\usepackage{url}
\usepackage[ruled,vlined]{algorithm2e}

\newenvironment{Figure}
{\par\medskip\noindent\minipage{\linewidth}}
{\endminipage\par\medskip}

%% file: abstract.tex
\begin{abstract}
    In this paper, we consider the problem of noiseless non-adaptive probabilistic group testing, in which the goal is high-probability recovery of the defective set.  We show that in the case of $n$ items among which $k$ are defective, the smallest possible number of tests equals $\min\{ C_{k,n} k \log n, n\}$ up to lower-order asymptotic terms, where $C_{k,n}$ is a uniformly bounded constant (varying depending on the scaling of $k$ with respect to $n$) with a simple explicit expression.  The algorithmic upper bound follows from a minor adaptation of an existing analysis of the Definite Defectives (DD) algorithm, and the algorithm-independent lower bound builds on existing works for the regimes $k \le n^{1-\Omega(1)}$ and $k = \Theta(n)$.  In sufficiently sparse regimes (including $k = o\big( \frac{n}{\log n} \big)$), our main result generalizes that of Coja-Oghlan {\em et al.} (2020) by avoiding the assumption $k \le n^{1-\Omega(1)}$, whereas in sufficiently dense regimes (including $k = \omega\big( \frac{n}{\log n} \big)$), our main result shows that individual testing is asymptotically optimal for any non-zero target success probability, thus strengthening an existing result of Aldridge (2019) in terms of both the error probability and the assumed scaling of $k$.
    
\end{abstract}

%% file: main_body.tex
\section{Introduction}

The group testing problem was originally studied in the context of testing blood samples for rare diseases \cite{Dor43}, with the key idea being to reduce the required number of tests via pooling.
Group testing has since found applications in communications \cite{Ant11}, information retrieval \cite{Cor05}, compressed sensing \cite{Gil08}, and most recently, COVID-19 testing \cite{Yel20}.

The problem is formally defined as follows: There are $n$ items $[n]=\sett{1,2,\ldots,n}$, a subset $S\subseteq [n]$ of which is defective, with $|S|=k$.  A number of tests are performed, each taking as input a subset of items, and returning positive if and only if the subset contains at least one defective item.  A group testing algorithm specifies the number of tests $T$, the items included in each test, and a decoder that returns an estimate $\hat{S}$ of the defective set given the test outcomes.   We are interested in the required number of tests to attain asymptotically vanishing error probability, i.e., $\lim_{n \to \infty} \PP[\hat{S}\ne S] = 0$. 


We focus on the non-adaptive setting, in which all tests must be specified prior to observing any outcomes; this is often highly desirable in applications, since it permits the tests to be implemented in parallel.  In this setting, the tests can be represented as a test matrix $\bb{X} \in \{0,1\}^{T \times n}$, where the $(i,j)$-th entry is 1 if and only if the $i$-th test contains the $j$-th item.  The test outcomes are then given by the element-wise ``OR'' of the $k$ columns corresponding to the defective items.  Mathematically, the $i$-th outcome is given by
\begin{align}
    Y^{(i)}=\bigvee_{j\in S}X_j^{(i)}, \label{eq:test_outcome_formula}
\end{align}
where $X_j^{(i)}$ is the $(i,j)$-th entry of $\bb{X}$.

We place a random model on the defective set $S$.  Throughout the majority of the paper, we assume that each item is included in $S$ (i.e., is defective) independently with some probability $p$ (possibly depending on $n$), referred to as the {\em prevalence}.  We also consider a closely-related model in which $k$ is fixed, and $S$ is a uniformly random subset of $[n]$ of cardinality $k$.  We refer to these models as the {\em i.i.d.~prior} and {\em combinatorial prior} respectively.  The two are closely related, since under the i.i.d.~model we have $k = np(1+o(1))$ with probability approaching one as long as $np = \omega(1)$.  See \cite[Sec.~1.A]{Ald19} for a more detailed summary of the connections between these models.

Throughout the paper, we make the mild assumption that $p \le \frac{1}{2}$ (i.i.d.~prior) or $k \le \frac{n}{2}$ (combinatorial prior).  Otherwise, the problem fails to be sparse, and it is already well-established that individual (one-by-one) testing is optimal even when adaptivity is allowed \cite{ungar1960,riccio2000sharper,aldridge2019rates}.  In addition, our analysis applies essentially unchanged when this factor of $\frac{1}{2}$ is replaced by any fixed constant less than one.

\section{Existing Results and Contributions} \label{sec:related}

Here we state the most relevant existing results on probabilistic non-adaptive group testing, and state our own main results in the context of these existing ones.  For consistency with the vast majority of existing works, we express the previously-known results in terms of $k$ and $n$, corresponding to the combinatorial prior.  However, the same results apply under the i.i.d.~prior when $k$ is replaced by $\kbar := np$ throughout, under the mild assumption that $\kbar \to \infty$ as $n \to \infty$.

\subsection{Main Existing and New Results}

A simple counting-based (or entropy-based) argument reveals that the number of tests for the high-probability recovery of $S$ must satisfy $T \ge (1-o(1)) \log_2 {n \choose k}$, or more simply $T = \Omega\big(k \log \frac{n}{k}\big)$ \cite{Mal78,Cha11,Ati12,Bal13}.  Moreover, this scaling is order-optimal in the following widely-considered regimes:
\begin{itemize}
    \item If $k \le n^{1-\Omega(1)}$, then we have $k \log \frac{n}{k} = \Theta(k \log n)$, and thus, the lower bound matches the ubiquitous $O(k \log n)$ upper bound obtained via random testing \cite{Fre75,Mal78,Cha11,Ald14a} or certain explicit designs \cite{Ina19}.
    \item If $k = \Theta(n)$, then we have $k \log \frac{n}{k} = \Theta(n)$, and thus, the lower bound matches the trivial $O(n)$ upper bound corresponding to testing each item individually.
\end{itemize}
While these observations cover the majority of scaling regimes, there remain ``mildly sublinear'' regimes in which the existing upper and lower bounds do not match, namely, $k = \Theta\big( \frac{n}{f(n)} \big)$ for any $f(n)$ satisfying $f(n) = \omega(1)$ and $f(n) = o(n^c)$ for all $c > 0$.  A notable example of such a regime is $k = \frac{n}{ {\rm poly}(\log n) }$.  Our first main result, stated below, closes this gap by showing that the correct scaling is always $\Theta( \min\{ k \log n, n\} )$.

Before stating the result, we introduce the following threshold:
\begin{equation}
    T^*(n,k) = \max\Big\{ k \log_2\frac{n}{k}, \frac{k \log_2 k}{\ln 2}\Big\}. \label{eq:T*}
\end{equation}
While the above discussion focuses on scaling laws, recent refined analyses \cite{Joh16,Ald18,Coj19,Coj19a} have nailed down the precise constants in the above-mentioned scaling regimes:
\begin{itemize}
    \item When $k \le n^{1-\Omega(1)}$, the optimal threshold for non-adaptive group testing is $T^*$ \cite{Coj19a}.  Specifically, there exists a strategy using $T \le (1+\epsilon)T^*$ tests that succeeds with probability approaching one and has decoding time polynomial in $n$, whereas any algorithm requires $T \ge (1-\epsilon)T^*$ to have a success probability bounded away from zero.
    \item When $k = \Theta(n)$, the optimal threshold for non-adaptive group testing is $n$ \cite{Ald18}.  Specifically, with $T = n$ one can trivially use one-by-one testing, whereas any strategy attaining success probability arbitrarily close to one must have $T \ge n-1$.\footnote{The subtraction of one is merely due to the fact that under the combinatorial prior, knowing the status of $n-1$ items also implies knowing the status of the remaining item.}
\end{itemize}
Based on these results, a reasonable guess is that the optimal threshold for group testing in general scaling regimes is $\min\{T^*(n,k),n\}$ (which scales as $\Theta( \min\{ k \log n, n\} )$, consistent with the above discussion).  Our main result, stated as follows, reveals that this is indeed the case.

\begin{thm} \label{thm:main1}
    In the non-adaptive group testing problem with $n$ items and prevalence $p$ (possibly depending on $n$) under the i.i.d.~prior with $p = \omega\big(\frac{1}{n}\big)$ and $p \le \frac{1}{2}$,\footnote{The assumption $p = \omega\big(\frac{1}{n}\big)$ ensures that $\kbar = np = \omega(1)$, and hence the number of defectives concentrates around $\kbar$.  The theorem cannot be true as stated when $p \le O\big(\frac{1}{n}\big)$; for example, even the trivial strategy of declaring every item non-defective has $\Omega(1)$ probability of succeeding, in contrast to the second part of the theorem.  In addition, our main novel contribution is handling the significantly denser regime $\kbar = n^{1-o(1)}$.  } we have the following for any $\epsilon > 0$:
    \begin{itemize}
        \item There exists a test design and polynomial-time decoding algorithm using $T \le \min\{ (1+\epsilon)T^*(n,np), n\}$ tests and having a success probability approaching one as $n \to \infty$.
        \item Any group testing strategy having success probability bounded away from zero as $n \to \infty$ must use at least $T \ge (1-\epsilon)\min\{ T^*(n,np), n \}$ tests. 
    \end{itemize}
\end{thm}

While we find it most convenient to establish this result for the i.i.d.~prior, we can use known connections between the two priors to establish the following analog for the combinatorial prior.

\begin{cor} \label{cor:main1}
    In the non-adaptive group testing problem with $n$ items and $k \le \frac{n}{2}$ defectives under the combinatorial prior, we have the following for any $\epsilon > 0$:
    \begin{itemize}
        \item There exists a test design and polynomial-time decoding algorithm using $T \le \min\{ (1+\epsilon)T^*(n,k), n\}$ tests and having a success probability approaching one as $n \to \infty$.
        \item Any group testing strategy having success probability bounded away from zero as $n \to \infty$ must use at least $T \ge (1-\epsilon)\min\{ T^*(n,k), n \}$ tests. 
    \end{itemize}
\end{cor}

These results not only show that optimal non-adaptive group testing requires $T = \Theta(\min\{k \log n,n\})$ tests in general sparsity regimes, but also provide the precise underlying constants.

The algorithmic upper bounds in Theorem \ref{thm:main1} and Corollary \ref{cor:main1} are already known in the regime $k \le n^{1-\Omega(1)}$ \cite{Joh16,Coj19a}, and follow trivially from one-by-one testing when $\min\{ (1+\epsilon)T^*, n\} = n$.  Hence, it suffices to establish success using $(1+\epsilon)T^*$ tests in the regime $k = n^{1-o(1)}$.  Fortunately, although the analysis of the Definite Defectives (DD) algorithm in \cite{Joh16} was formally only stated for $k \le n^{1-\Omega(1)}$, the analysis can be adapted to the regime $k = n^{1-o(1)}$ with only minor modifications.  We detail the required changes in Appendix \ref{app:dd}.

As for the algorithm-independent lower bounds, with the regime $k \le n^{1-\Omega(1)}$ having been solved in \cite{Coj19a}, we can again focus on the regime $k = n^{1-o(1)}$ (including $k = \Theta(n)$).  In this case, we were unable to directly infer the desired result from \cite{Coj19a}, and we thus provide a detailed proof in Section \ref{sec:proofs}, though we still naturally re-use the main tools and ideas proposed in \cite{Coj19a}.


{\bf Specialization to dense regimes.} As hinted above, when the $\min\{T^*,n\}$ term is attained by $n$, our results indicate that one-by-one testing is asymptotically optimal.  This is consistent with the above-mentioned result of \cite{Ald18}, but also strengthens it in two ways:
\begin{itemize}
    \item Individual testing is not only asymptotically optimal when the goal is to succeed with probability approaching one, but also when the goal is attaining {\em any} strictly positive target success probability.
    \item Individual testing is not only asymptotically optimal when $k = \Theta(n)$, but also when $k = \omega\big( \frac{n}{\log n} \big)$, or even more generally, when $k > \frac{n \ln 2}{\log_2 n}$.
\end{itemize}
On the other hand, it is worth noting that our result only indicates failure when $T < (1-\epsilon)n$, whereas that of \cite{Ald18} handles the more general scenario $T < n-1$.  This distinction is necessary when establishing high-probability failure and/or handling the regime $k = o(n)$, since otherwise one could consider a strategy that (e.g.) tests the first $n-2$ items one-by-one and then guesses the remaining two to be non-defective.

{\bf Note on partially concurrent work.} In the initial version of our work, we focused only on the lower bound, and provided a weaker result with an unspecified coefficient to the $k \log n$ term in the $\min\{ k \log n, n \}$ scaling.  After releasing the initial version, the important case of $k = \Theta\big( \frac{n}{\log n}\big)$ was studied in more detail in \cite{Flo21}, giving upper and lower bounds with explicit constants.  The updated version of our work was developed in parallel with \cite{Flo21}, and establishes the precise constants.  Our upper bound in fact matches that of \cite{Flo21} (and is proved similarly), whereas a refinement of the main proof technique is needed to obtain our tight lower bound (see the stopping condition of Step 4(a), Procedure \ref{alg:fullalg}).

\subsection{Further Existing Results}

Before proceeding, we provide a brief summary of some further existing works.  Since these are less directly related to our work, we omit the details, and refer the reader to \cite{Du93,Ald19} for more detailed surveys.


For certain variants of group testing, the optimal number of tests is $\Theta\big( k \log \frac{n}{k}\big)$, as opposed to $\Theta(\min\{k \log n,n\})$ under the setup we consider.  Specifically, two notable cases with scaling $\Theta\big( k \log \frac{n}{k}\big)$ are (i) the adaptive setting, in which each test can be designed based on previous outcomes \cite{Hwa72,aldridge2019rates,aldridge2020conservative}, and (ii) the approximate recovery criterion, in which $\Theta(k)$ false positives and $\Theta(k)$ false negatives are allowed in the reconstruction \cite{Sca15b,Sca17}.  

In contrast to the noiseless setting that we consider in this paper, in the {\em noisy} setting, the number of tests is at least $\Omega(k \log n)$ even if $k \log n \gg n$, and even if adaptivity is allowed \cite{Sca18}.

Finally, while the focus of our work is on high-probability recovery, extensive results have been established for the stronger guarantee of {\em uniform recovery}, i.e., a single test matrix that uniquely recovers any defective set of cardinality at most $k$, without allowing any error probability (e.g., see \cite{Kau64,Dya82,Du93,Che13a} and the references therein).  This stronger guarantee comes at the price of requiring significantly more tests, with a quadratic dependence on $k$ instead of a linear dependence.  In addition, the associated proof techniques are very different.

\section{Proofs of Algorithm-Independent Lower Bounds} \label{sec:proofs}

We first consider Theorem \ref{thm:main1} regarding the i.i.d.~prior, and then turn to Corollary \ref{cor:main1} regarding the combinatorial prior.  

\subsection{Proof of the Lower Bound for Theorem \ref{thm:main1}} \label{sec:pf_thm}

Our analysis builds on the ideas of \cite{Ald18,Coj19a}, both of which identify {\em totally disguised} items (see Definition \ref{def:disguised} below) whose defectivity status can be flipped without changing the test outcomes.  In \cite{Ald18}, one such item suffices for attaining the weak converse (i.e., $\PP[\hat{S} \ne S] \not\to 0$) in the linear regime.  To obtain a stronger statement of the form $\PP[\hat{S} \ne S] \to 1$ and also handle sublinear sparsity regimes, we follow the idea from \cite{Coj19a} of identifying {\em many} such items.

Specifically, we follow the high-level steps of \cite{Coj19a} and utilize certain auxiliary results therein, but modify the details in order to handle the regime $k = n^{1-o(1)}$ instead of $k \le n^{1-\Omega(1)}$.  The key idea is to identify many items that are disguised {\em independently of one another}.  We then apply an auxiliary result of \cite{Ald18} (see Lemma \ref{lem:Aldridge} below) along with some ``clean-up'' steps to ensure that its assumptions remain valid each time it is invoked.

In the following, we let $q=1-p$ for convenience.  The following useful definition was introduced in \cite{Ald18}.

\begin{defn} \label{def:disguised}
    \cite{Ald18} 
    We say that an item $i$ is \textit{disguised} in test $t$ if at least one of the other items in the test is defective. We say that an item is \textit{totally disguised} if it is disguised in every test it is included in. Let $D_i$ denote the event that item $i$ is totally disguised.
\end{defn}

It is noted in \cite{Ald18} that if an item is totally disguised, then it remains totally disguised even if it is changed from defective to non-defective or vice versa.  Thus, under the i.i.d.~prior, the tests do not reveal any information about that item's defectivity status, and we have the following.

\begin{lem} \label{lem:cond_disguised}
    {\em (Implicit in \cite{Ald18} and \cite[Sec.~3]{Coj19a})}
    For any given test matrix $\Xv$, and a defective set $S$ generated according to the i.i.d.~prior, we have the following: Conditioned on a given item $i$ being totally disguised, that item is defective with conditional probability $p$ (i.e., the same as the prior defectivity probability).
\end{lem}

It follows that for any totally disguised item, the best the algorithm can do is choose the more likely outcome, and succeed with probability $\max\sett{p,1-p} = 1-p$ (recalling that we focus on the case that $p \le \frac{1}{2}$).

The following result from \cite{Ald18} is crucial for characterizing the probability of items being totally disguised.

\begin{lem} \label{lem:Aldridge}
    {\em \cite[Eq.~(1)]{Ald18}}
    Define $\mathcal{L}(p)=\min_{x=2,3,\ldots,n}x\ln(1-q^{x-1})$, where $q = 1-p$. If the test design $\bb{X}$ has no tests with 0 or 1 items, then 
    \begin{align}
        \frac{1}{n}\sum_{i=1}^n\ln \pr[D_i]\geq\frac{T}{n}\cdot \mathcal{L}(p) \label{eq:AldBound}
    \end{align}
    Hence, there exists an item $i$ with $\ln \pr[D_i]\geq\frac{T}{n}\cdot \mathcal{L}(p)$.
\end{lem}

At a high level, this lemma is proved by directly calculating $\pr[D_i]$ in terms of the $i$-th test size (which $x$ plays the role of in the definition of $\mathcal{L}(p)$), then averaging over the resulting log-values and applying some simple lower bounding techniques.

In the linear regime (i.e., $k = \Theta(n)$), one has $\mathcal{L}(p) = \Theta(1)$, and consequently, Lemma \ref{lem:Aldridge} directly implies that the success probability is bounded away from one, after removing all tests with 0 or 1 items \cite{Ald18}.  
More generally, it is natural to ask whether there are, in fact, {\em many} items $i$ with $\ln \pr[D_i]$ close to the the right-hand side of \eqref{eq:AldBound} \cite{Coj19a}. If we can find a ``large'' set $W$ of such items such that these items are totally disguised independently from each other, then we may apply standard binomial distribution concentration bounds to conclude that many totally disguised items exist, with high probability.


Following \cite{Coj19a}, we interpret the testing strategy as a bipartite graph $G_{\bb{X}}$ in which there is a vertex $v_i$ for each item $i$ and a vertex $v_t$ for each test $t$, with an edge between $v_i$ and $v_t$ if item $i$ is placed in test $t$. 
Before constructing the desired set (denoted by $W$), we present two simple lemmas (which are analogous to \cite[Lemmas 3.7 and 3.8]{Coj19a}) and two subroutines that will be useful.  

\begin{lem}
    \label{nobigtests}
    Let $z=\frac{2}{\ln\frac{1}{q}}$, and suppose that $T \le n$. Then, the probability that there exists a negative test containing more than $z\ln n$ items is at most $\frac{1}{n}$.
\end{lem}

\begin{proof} 
    Recalling that $q = 1-p$, a given test containing at least $z\ln n$ items is negative with probability at most 
    \begin{align}
        q^{z\ln n}=e^{z\ln q\ln n}=\frac{1}{n^2}
    \end{align}
    by the definition of $z$.  Since $T \le n$, a union bound yields the desired result.
\end{proof}

We henceforth assume that no test contains more than $z\ln n$ items, since Lemma \ref{nobigtests} implies that the decoder may declare all such tests to be positive without increasing the error probability by more than $\frac{1}{n}\to0$.  

\begin{lem}
    \label{fewbigitems}
    Fix $\xi > 0$, and define an item to be {\em very-present} if it appears in more than $n^{\xi}$ tests. If $T \le n$ and no test contains more than $z\ln n$ items, then there are no more than $zn^{1-\xi}\ln n$ very-present items.
\end{lem}

\begin{proof} 
    We count the number $P$ of pairs $(i,t)$ such that item $i$ is in test $t$. By assumption, $P \leq Tz\ln n \le nz\ln n$. Letting $n_{\mathrm{vp}}$ be the number of very-present items, it follows that $n_{\mathrm{vp}}n^{\xi}\leq P< nz\ln n$, and rearranging yields the desired result.
\end{proof}

\begin{algorithm}[t]
    \caption*{ \manuallabel{sub}{3.1} \textbf{Subroutine 3.1:} $\texttt{Clean}(\bb{X})$. \label{alg:clean} }
    
    
    \begin{enumerate}
        \item[1.] Identify the set of tests $T_{\leq1}$ containing 0 or 1 items, and the set of items $I$ contained in at least one test in $T_{\leq1}$.
        \item[2.] Return $\bb{X}_{\geq2}$, defined to be $\bb{X}$ with the rows and columns indexed by $T_{\leq1}$ and $I$ removed.
    \end{enumerate}
\end{algorithm}

\begin{algorithm}[t]
    \caption*{ \manuallabel{subr}{3.2} \textbf{Subroutine 3.2:} $\texttt{Extract}(\bb{X},W)$. \label{alg:extract} }
    
    
    
    \begin{enumerate}
        \item[1.] Let $\widetilde{D}_i$ be the event that $i$ is totally disguised with respect to $\bb{X}$. Let the item with the highest $\pr[\widetilde{D}_i]$ be denoted by $i_0$, and set $W_{\rm next} = W \cup \{i_0\}$.
        \item[2.] Let $T_{\textrm{close}}$ and $I_{\textrm{close}}$ denote the sets of tests and items within distance at most 4 from $i_0$ in $G_{\bb{X}}$.
        \item[3.] Set $\bb{X}_{\rm pruned}$ to be $\bb{X}$ with the rows and columns indexed by $T_{\textrm{close}}$ and $I_{\rm close}$ removed.
        \item[4.] Return $(\bb{X}_{\textrm{pruned}},W_{\rm next})$
    \end{enumerate}
\end{algorithm}

We are now in a position to describe the construction of the desired set $W$, namely, a set of items that are disguised independently of one another.  To establish a hardness result, we would like to ensure that the size of $W$ and the probability of each $i \in W$ being disguised are both large enough, so that the resulting error probability is high.

Towards achieving this goal, we introduce Subroutines \ref{alg:clean} and \ref{alg:extract}.  \texttt{Clean} removes all tests with 0 or 1 items, allowing us to apply Lemma \ref{lem:Aldridge}, and \texttt{Extract} adds an item to $W$.  Both will be called multiple times in the construction of $W$, and their calls will reduce the effective $T$ and/or $n$.


\begin{algorithm}[t]
    \caption*{ \manuallabel{algo}{3.1} \textbf{Procedure 3.1:} $\texttt{ConstructSet}(\bb{X})$. \label{alg:fullalg} }
    
    
    
    \begin{enumerate}
        \item[1.] Let $G_0=G_{\bb{X}}$, and let $(n,T)$ be the number of items and tests in $\bb{X}$. Remove all very-present items from $G_0$ to obtain $G_1$. Let $G=G_1$.
        \item[2.] Initialize $W_0=\emptyset$, $i=1$.
        \item[3.] Set $\bb{X}_i\leftarrow$ test design represented by $G_i$. Set $\bb{X}_{\mathrm{tmp},i}\leftarrow\texttt{Clean}(\bb{X}_i)$, and let $(n_i,T_i)$ be the corresponding number of items and tests in $\bb{X}_{\mathrm{tmp},i}$.
        \item[4.] Perform the following:
        \begin{itemize}[leftmargin=3ex]
            \item[(a)] If $n_i > 0$ and $\frac{T_i}{n_i} \le (1+\xi) \frac{T}{n}$, then set $(\bb{X}_{i+1},W_{i}) \leftarrow\texttt{Extract}(\bb{X}_{\mathrm{tmp},i},W_{i-1})$, $G_{i+1}\leftarrow G_{\bb{X}_{i+1}}$, and $i\leftarrow i+1$, and return to Step 3. 
            \item[(b)] Otherwise, terminate the procedure and return $W = W_{i-1}$.
        \end{itemize}
    \end{enumerate}
\end{algorithm}

The full procedure for constructing $W$ is described in Procedure \ref{alg:fullalg}, which depends on a generic constant $\xi > 0$; although its use in step 4(a) is not directly related to its used in Lemma \ref{fewbigitems}, we find it sufficient to use the same constant in both cases.  
To justify step 1, we momentarily imagine that there exists a ``genie'' that tells the decoder the identity of the very-present items. Let the test results for $G_0$ and $G_1$ be $\bb{y}_0$ and $\bb{y}_1$ respectively; then, knowing $\bb{X}$, we see that $\bb{y}_0$ can be derived from $\bb{y}_1$ and the genie information. If we can prove that the error probability tends to one even with the help of the genie (and knowing $\bb{y}_1$), then it certainly tends to one without it, so step 1 is justified. After step 1, each item is contained in at most $n^{\xi}$ tests.

Let $w_i$ denote the $i$-th item placed in $W$. Let $D_{w_i}$ be the event that $w_i$ is totally disguised with respect to $\bb{X}_1$, and let $\widetilde{D}_{w_i}$ be the event that $w_i$ is totally disguised with respect to $\bb{X}_{{\rm tmp},i}$ (see Procedure \ref{alg:fullalg} for the definitions of  $\bb{X}_1$ and $\bb{X}_{{\rm tmp},i}$).  
Since the totally disguised event $D_{w_i}$ only depends on the 2-neighborhood of $w_i$ in $G_1$, and the 2-neighborhoods of items in $W$ are pairwise disjoint by construction (due to the $\texttt{Extract}$ subroutine), the events $\sett{D_{w}:w\in W}$ are independent (this independence property for nodes having distance greater than $4$ was also used in \cite{Coj19a}). 

Next, we state the following simple lemma relating the events $D_{w_i}$ and $\widetilde{D}_{w_i}$, both of which represent events of being totally disguised, but with  respect to different test matrices.

\begin{lem} \label{lem:EtoD}
    Under the preceding setup, we have $\pr[D_{w_i}] \geq \pr[\widetilde{D}_{w_i}]$. 
\end{lem}
\begin{proof}
    In each \texttt{Clean}/\texttt{Extract} step, whenever we remove a test, we remove all of its items. It follows that $w_i$ is contained in $\textit{the same tests}$ in $\bb{X}_1$ and $\bb{X}_{\mathrm{tmp},i}$, except that each such test in $\bb{X}_{\mathrm{tmp},i}$ has \textit{fewer items}.  Since a disguised item always remains disguised when further items are added to its tests, it follows that $\widetilde{D}_{w_i}$ implies $D_{w_i}$.  
\end{proof}

In addition, we have the following lower bound on $|W|$, the total number of extracted items.  Here and subsequently, we recall that to prove Theorem \ref{thm:main1}, it suffices to consider the regime $p = n^{-o(1)}$, since for any smaller $p$ (i.e. $p = n^{-\Omega(1)}$), Theorem \ref{thm:main1} was already established in \cite{Coj19a}.  



\begin{lem} \label{lem:num_extracted}
    Under the preceding setup, if $p = n^{-o(1)}$ and $T \le (1-\epsilon)n$, then the size of the set $W$ returned by Procedure \ref{alg:fullalg} satisfies the following:
    \begin{equation}
        |W|\, \ge n^{1-3\xi}. \label{eq:W_asymp}
    \end{equation}
\end{lem}
\begin{proof}
    We first count the number of removed items as follows:
    \begin{itemize}
        \item No more than $T$ items alone in some test are removed by \texttt{Clean}.
        \item Lemma \ref{fewbigitems} implies that we removed at most $zn^{1-\xi}\ln n$ very-present items, and this scales as $o(n)$ due to the fact that $z = \frac{2}{\ln\frac{1}{1-p}} = \Theta\big(\frac{1}{p}\big) = n^{o(1)}$ (by the assumption $p = n^{-o(1)}$).
        \item By the assumption stated following Lemma \ref{nobigtests} and the removal of very-present items, each call to \texttt{Extract} removes at most $z^2n^{2\xi}\ln^2 n$ items.
    \end{itemize}
    We now argue by contradiction that \eqref{eq:W_asymp} must hold.  Suppose to the contrary that Procedure \ref{alg:fullalg} terminates at some iteration $i^* \le n^{1-3\xi}$.  Then, the above calculations imply that
    \begin{align}
        n_{i^*} &\ge n - T - zn^{1-\xi}\ln n- n^{1-3\xi} \cdot  z^2n^{2\xi}\ln^2n \label{eq:n_i_0} \\
        &\ge \epsilon n - o(n), \label{eq:n_i}
    \end{align}
    where we used the fact that $T \le (1-\epsilon)n$ and $z = n^{o(1)}$.  This means that the stopping condition met in step 4(a) cannot have been $n_i$ reaching zero, so it must have been $\frac{T_{i^*}}{n_{i^*}}$ exceeding $(1+\xi) \frac{T}{n}$.
    
    While \eqref{eq:n_i} indicates that the majority of items could eventually be removed in principle, this is only due to the subtraction of $T$ in \eqref{eq:n_i_0}; the other two terms behave as $o(n)$, and we conclude that the removal of very-present items and the calls to \texttt{Extract} collectively only remove $o(n)$ items.  Any further removal of items can only be due to the tests containing one item in \texttt{Clean}; removing these causes $T_i$ and $n_i$ to be reduced by the {\em same amount}.  However, as long as $\frac{T_i}{n_i} < 1$ (which holds by assumption for $i=0$, and subsequently for all $i \le i^*$ due to the stopping condition), reducing $T_i$ and $n_i$ by the same amount can {\em only make the ratio smaller} (i.e., $\frac{T_i - c}{n_i - c} \le \frac{T_i}{n_i}$ for any $c \in [0,T_i]$).
    
    More formally, suppose that up to index $i^*$, a total of $c$ items and tests are removed due to tests containing a single item, and a total of $c'$ items are removed for the other reasons mentioned above.  Then, we have
    \begin{equation}
        \frac{T_{i^*}}{n_{i^*}} \le \frac{ T - c }{ n - c - c' }. \label{eq:ratio}
    \end{equation}
    As established above, we have $0 \le c \le T \le n(1-\epsilon)$ and $c' = o(n)$.  However, since $\frac{T}{n} < 1$ by assumption, substituting these findings into \eqref{eq:ratio} reveals that $\frac{T_{i^*}}{n_{i^*}} \le \frac{T}{n}(1+o(1))$, which gives the desired contradiction to the stopping condition $\frac{T_{i^*}}{n_{i^*}} > (1+\xi) \frac{T}{n}$.

\end{proof}

Recall that all of the extracted items have independent totally disguised events, each with probability lower bounded according to Lemma \ref{lem:Aldridge}.  
We need to consider applying this lemma with possibly smaller choices of $T$ and $n$ than the original values (namely, $T_i$ and $n_i$), but the stopping condition in step 4(a) of Procedure \ref{alg:fullalg} ensures that $\frac{T_i}{n_i} \le (1+\xi) \frac{T}{n}$.  
As a result, Lemmas \ref{lem:Aldridge} and \ref{lem:EtoD} guarantee for any extracted item $i$ that 
\begin{align}
    \pr[D_i]\geq\exp\bigg( \frac{T(1+\xi)}{n}\cdot \mathcal{L}(p) \bigg), \label{eq:Di_bound1}
\end{align}
where we recall that $\mathcal{L}(p)=\min_{x=2,3,\ldots,n}x\ln(1-q^{x-1})$ with $q = 1-p$.  Note that $x\ln(1-q^{x-1}) < 0$, so this minimum is to be interpreted as ``most negative''. This minimum is characterized in the following lemma, which is similar to \cite[Claim 3.12]{Coj19a}.

\begin{lem} \label{lem:x_opt}
    For any $p \le \frac{1}{2}$ satisfying $p = n^{-o(1)}$, we have the following: (i) If $p = o(1)$, then $-\mathcal{L}(p) = \frac{(\ln 2)^2}{p} (1+o(1))$; (ii) If $p = \Theta(1)$, then  $-\mathcal{L}(p) = \Theta(1)$.
\end{lem}
\begin{proof}
    We provide a simple generalization of the argument from \cite[Claim 3.12]{Coj19a}, which focuses on the regime $k \le n^{1-\Omega(1)}$.
    We first write
    \begin{equation}
        -\mathcal{L}(p)=\max_{x=2,3,\ldots,n}x\ln\frac{1}{1-q^{x-1}}.
    \end{equation}
    This quantity is lower bounded by the argument corresponding to $x = \lceil \frac{1}{p}\rceil + 1$, which readily yields $\ln\frac{1}{1-q^{x-1}} = \Theta(1)$ and hence an $\Omega\big(\frac{1}{p}\big)$ lower bound on $-\mathcal{L}(p)$.
    
    For the upper bound, we note that for $x = o\big( \frac{1}{p} \big)$ we have $q^{x-1} = (1-p)^{x-1} = 1-\Theta(px)$, so the objective function behaves as $O(x \ln \frac{1}{px})$, which is  $o\big(\frac{1}{p}\big)$ (since $px \ln \frac{1}{px} \to 0$ as $px \to 0$).  On the other hand, if $x = \omega\big(\frac{1}{p}\big)$, then $q^{x-1} = (1-p)^{x-1} \to 0$, so the objective behaves as $O(x (1-p)^{x-1}) = O(x e^{-px})$, which is $o\big(\frac{1}{p}\big)$  (since $px e^{-\Theta(px)} \to 0$ as $px \to \infty$).  Hence, the optimal choice of $x$ must scale as $\Theta\big(\frac{1}{p}\big)$, and in this case, we have $\ln\frac{1}{1-q^{x-1}} = \Theta(1)$, yielding an $O\big(\frac{1}{p}\big)$ upper bound on $-\mathcal{L}(p)$.
    
    The second part of the lemma follows immediately, whereas for the first part, a refined analysis is needed.  For $p = o(1)$ and $x = \Theta\big( \frac{1}{p} \big)$, we have
    \begin{align}
        x\ln\frac{1}{1-(1-p)^{x-1}}
        &= x\ln\frac{1}{1-e^{-px}(1+O(p))} \\
        &= -x\ln(1-e^{-px}) + O(px)
    \end{align}
    by standard Taylor expansions.  Since $O(px) = O(1)$ is asymptotically negligible compared to the $\Theta\big( \frac{1}{p} \big)$ scaling derived above, it suffices to consider maximizing the first term.  As noted in \cite{Coj19a}, we can define $d = px$ and write this term as $\frac{1}{p}\big( - d \ln(1-e^{-d}) \big)$, and it is a simple differentiation exercise to verify that $-d \ln(1-e^{-d})$ is maximized at $d = \ln 2$, with maximum value $(\ln 2)^2$.  While the corresponding choice $x = \frac{d}{p}$ may not be integer-valued, the effect of rounding is asymptotically negligible for $p = o(1)$ by the continuity of the function $-d \ln(1-e^{-d})$.
\end{proof}

Combining Lemma \ref{lem:x_opt} with \eqref{eq:Di_bound1}, we obtain for some $c_p = \Theta(1)$ that
\begin{align}
    \pr[D_i]\geq\exp\bigg( - \frac{(1+\xi) c_p T}{n p} \bigg), \label{eq:Di_bound2}
\end{align}
and moreover, when $p = o(1)$ we specifically have $c_p = (\ln 2)^2 (1+o(1))$.

Since the events $\{D_i\}_{i \in W}$ are mutually independent by construction, and $|W| \ge n^{1-3\xi}$ according to Lemma \ref{lem:num_extracted}, we deduce that the number of totally disguised items is stochastically dominated by ${\rm Binomial}( n^{1-3\xi}, e^{-\frac{(1+\xi) c_p T}{ np }} )$.  In particular, the average number of totally disguised items is
\begin{equation}
    n^{1-3\xi} e^{-\frac{(1+\xi) c_p T}{ np }},
\end{equation}
and by simple re-arrangements, this is lower bounded by $n^{\xi}$ whenever
\begin{equation}
    T \le \frac{np(1-4\xi)}{(1+\xi) c_p } \ln n. \label{eq:T_bound2}
\end{equation}
Thus, by the multiplicative form of the Chernoff bound, the actual number is at least $N_{\min} := \frac{1}{2} n^{\xi}$ with probability approaching one when \eqref{eq:T_bound2} holds.

By Lemma \ref{lem:cond_disguised} and the assumption $p \le \frac{1}{2}$, for any item that is disguised, the optimal algorithm can do no better than declare it to be non-defective, and the resulting probability of being correct is at most
\begin{equation}
    (1-p)^{N_{\min}} \le e^{-p N_{\min}} = e^{-\frac{p}{2} n^{\xi} } = o(1),
\end{equation}
where the last step follows from the assumption $p =  n^{-o(1)}$.  Thus, we have proved that $\PP[\hat{S} = S] = o(1)$ whenever $T \le (1-\epsilon) n$ and \eqref{eq:T_bound2} holds.

When $p = \Theta(1)$ (or more generally $p = \omega\big( \frac{1}{\log n}\big)$), the stricter of these two conditions is $T \le (1-\epsilon) n$, and the constant $c_p$ in \eqref{eq:T_bound2} is inconsequential.  On the other hand, when $p = o(1)$, we have established that $c_p = (\ln 2)^2 (1+o(1))$.  Since $\xi$ can be arbitrarily small, it follows that \eqref{eq:T_bound2} reduces to $T \le (1-\epsilon') \frac{np}{ (\ln 2)^2 } \ln n$ for arbitrarily small $\epsilon' > 0$.  Finally, since $\frac{\ln n}{\ln 2} = \log_2 n$, and the assumption $p = n^{-o(1)}$ implies that $\log_2 n = (\log_2 k)(1+o(1))$, we obtain the desired threshold corresponding to the second term of $T^*$ in \eqref{eq:T*}.

\subsection{Proof of the Lower Bound for Corollary \ref{cor:main1}} \label{sec:pf_cor}

We utilize an approach from \cite[Lemma 3.6]{Coj19a} for transferring the key auxiliary results on the number of disguised items from the i.i.d.~prior to the combinatorial prior.  Despite the high level of similarity, we provide the main details for completeness.

The idea is to show that with too few tests, the number of totally disguised defectives and totally disguised non-defectives both grow unbounded with high probability.  When this occurs, interchanging the statuses among these items would not impact the test results, and hence, there exist an unbounded number of candidate defective sets of cardinality $k$ consistent with the test outcomes.  The decoder cannot do any better than guess one of these at random, failing with high probability. This intuition is easily made precise \cite{Coj19a}, giving the following.

\begin{lem} \label{lem:cond_error}
    {\em \cite[Facts 3.1 and 3.3]{Coj19a}}
    Under the combinatorial prior, the conditional error probability of any group testing strategy given that there are $\tilde{n}_0$ totally disguised non-defectives and $\tilde{n}_1$ totally disguised defectives is at least $1 - \frac{1}{\tilde{n}_0 \tilde{n}_1}$.   In particular, if $\tilde{n}_0 = \omega(1)$ and $\tilde{n}_1 = \omega(1)$, then the conditional error probability is $1-o(1)$.
\end{lem}

Consider the combinatorial prior with $n^{1-o(1)} \le k \le \frac{n}{2}$, where the condition $k \ge n^{1-o(1)}$ is safe to assume since Corollary \ref{cor:main1} is already well-known when $k \le n^{1-\Omega(1)}$ (see Section \ref{sec:related}).  We consider generating $S$ according to the following procedure:
\begin{enumerate}
    \item Let $S_{0} \subseteq [n]$ include each item independently with probability $p_0 = \frac{k - \sqrt{k} \ln n}{n}$.  That is, $S_0$ follows the i.i.d.~prior with parameter $p_0$.
    \item Form $S$ by adding $\max\{k - |S_0|, 0\}$ elements of $[n] \setminus S_0$ to $S_0$, chosen uniformly at random.
\end{enumerate}
By the symmetry of this construction, conditioned on the event $|S_0| \le k$, the resulting set $S$ is indeed distributed according to the combinatorial prior.  While $|S_0| >  k$ has a non-zero probability, for the purposes of proving a converse, we can simply assume that this event always leads to successful recovery.  Since we assume that $k \ge n^{1-o(1)}$, a simple concentration argument (e.g., the Chernoff bound or central limit theorem) gives with probability $1-o(1)$ that
\begin{equation}
    k - 2\sqrt{k} \ln n \le |S_0| \le k, \label{eq:conc}
\end{equation}
so the resulting contribution to the success probability is asymptotically negligible.

We now introduce the terminology that an item $i$ is {\em totally disguised in the first step} if the defectives from $S_0$ alone are enough to disguise $i$ in every test it is included in.  Clearly, being totally disguised in the first step is sufficient for being totally disguised after the second step, since the second step only involves marking more items as defective.

Hence, trivially, the number of totally disguised defective items only increases (or stays the same) after the second step.  The number of totally disguised non-defectives may in principle decrease due to non-defectives being changed to defective, but conditioned on \eqref{eq:conc}, any given non-defective is only changed with probability $O\big( \frac{\sqrt{k} \ln n}{n} \big) = o(1)$.  As a result, if there are $\omega(1)$ totally disguised non-defectives, the same still remains true with probability $1-o(1)$ after the second step.

Hence, in accordance with Lemma \ref{lem:cond_error}, it suffices to show that under the i.i.d.~prior with parameter $p_0 = \frac{k - \sqrt{k} \ln n}{n}$, the number of totally disguised defectives and totally disguised non-defectives both behave as $\omega(1)$ with probability $1-o(1)$.   Note that the assumption $n^{1-o(1)} \le k \le \frac{n}{2}$ ensures that $n^{-o(1)} \le p_0 \le \frac{1}{2}$, as was assumed in the later parts of Section \ref{sec:pf_thm}.

We already argued that when \eqref{eq:T_bound2} holds, the average number of totally disguised items is at least $n^{\xi}$.  Since $n^{-o(1)} \le p_0 \le \frac{1}{2}$, it follows that the average number of totally disguised defectives and totally disguised non-defectives are both at least $n^{\xi-o(1)}$ on average.  Again, the multiplicative form of the Chernoff bound implies the same with high probability, and we have the desired $\omega(1)$ scaling.  This establishes that the condition on $T$ from Theorem \ref{thm:main1} with $p_0$ in place of $p$ is necessary for attaining a success probability bounded away from zero, and since $np_0 = k(1+o(1))$ by definition, Corollary \ref{cor:main1} follows.

\section{Conclusion} \label{sec:conclusion}

We have proved that the optimal number of tests for probabilistic noiseless non-adaptive group testing is $\Theta(\min\{k \ln n,n\})$, as well as establishing the precise underlying constant factors.  This closes gaps exhibited by existing bounds in the in the case that $k$ is ``mildly'' sublinear in $n$, so that the optimal thresholds are now known for arbitrary scaling regimes.  Perhaps the main challenge remaining in this setting is to understand how the number of tests increases when the target error probability decreases to zero at a given rate depending on $n$.

%% file: appendix.tex
\section{Proofs of Algorithmic Upper Bounds} \label{app:dd}

Under the combinatorial prior, the algorithmic upper bound for Theorem \ref{thm:main1} in the regime $k \le n^{1-\Omega(1)}$ is proved in \cite{Joh16} using the following strategy:
\begin{itemize}
    \item {\bf Test design.} Generate the $T \times n$ test matrix $\Xv$ according to the {\em near-constant tests-per item} design: For each item $i = 1,\dotsc,n$, select $L = \lfloor \frac{T \ln 2}{k} \rfloor$ tests uniformly at random {\em with replacement}, and set the corresponding entries in the $i$-th column of $\Xv$ to one.
    \item {\bf Decoding algorithm.} Given the test outcomes, estimate the defective set using the {\em Definite Defectives (DD)} algorithm:
    \begin{itemize}[leftmargin=7ex]
        \item[(i)] Mark all items in negative tests as {\em definitely non-defective}, and all remaining items as {\em possibly defective} (PD);
        \item[(ii)] For any PD item appearing in some (necessarily positive) test without any other PD items, mark it as {\em definitely defective} (DD).
        \item[(iii)] Return the set of DD items as the final estimate.
    \end{itemize}
\end{itemize}
The main result of \cite{Joh16} states that when $k = \Theta(n^{\theta})$ with $\theta \in (0,1)$ and the number of tests satisfies
\begin{equation}
    T \ge \frac{ \max\{\theta,1-\theta\} }{\ln 2} \big(k \log_2 n\big) (1+\epsilon) \label{eq:T_old}
\end{equation}
for some $\epsilon > 0$, the resulting error probability approaches zero as $n \to \infty$.  Our goal is to generalize this result to denser sparsity regimes.

The condition \eqref{eq:T_old} ensures that $L = \lfloor \frac{T \ln 2}{k} \rfloor$ scales as $\omega(1)$.  Hence, the effect of rounding is negligible, in the sense that $L =  \frac{T \ln 2}{k} (1+o(1))$.  As in \cite{Joh16}, we subsequently work with the exact expression $L =  \frac{T \ln 2}{k}$ for notational convenience, since the $o(1)$ term does not affect the final result.

With the regime $k \le n^{1-\Omega(1)}$ having been handled in \cite{Joh16}, it suffices to consider $k = n^{1-o(1)}$.  In this regime, it holds that $T^*(n,k) = \frac{k \log_2 k}{\ln 2}$.  For convenience, we apply the fact that $\log_2 k = (\log_2 n)(1+o(1))$ (whenever $k = n^{1-o(1)}$), meaning that it suffices to show that the success probability approaches one when
\begin{equation}
    T \ge \frac{k \log_2 n}{\ln 2} (1+\epsilon). \label{eq:T_choice}
\end{equation}
Observe that this matches \eqref{eq:T_old}, but with the quantity 
\begin{equation}
    m =\max\{\theta,1-\theta\} \label{eq:m_def}
\end{equation}
replaced by $m=1$.

{\bf Analysis.} We start with the following bound which is central to the analysis of \cite{Joh16}, and is conveniently non-asymptotic so can can also be used here: For any defective $i \in S$, denoting the final estimate by $\Shat$, we have
\begin{align}
    \PP[i \notin \Shat] 
    &\le \underbrace{\sum_{w \in [w_-,w_+]} \PP[W^{(S\setminus i)} = w] \sum_{j=0}^L \PP[M_i = j | W^{(S\setminus i)} = w] \phi_j(1/w_-, g^* L)}_{:=\Psi_1} \nonumber \\
    &\qquad\qquad\qquad + \underbrace{\PP[W^{(S\setminus i)} \notin [w_-,w_+]]}_{:=\Psi_2} + \underbrace{\PP[G > g^* | W^{(S\setminus i)} \notin [w_-,w_+]]}_{:=\Psi_3}, \label{eq:three_terms}
\end{align}
where:
\begin{itemize}
    \item $W^{(S\setminus i)}$ denotes the number of (necessarily positive) tests containing at least one item in $S \setminus \{i\}$, i.e., a defective item differing from $i$;
    \item $M_i$ denotes the number of tests containing $i \in S$ and no other defectives;
    \item $G$ denotes the number of non-defectives that do not appear in any negative tests;
    \item $w_-$ and $w_+$ are arbitrary thresholds, but should be chosen to ensure that $W^{(S\setminus i)} \in [w_-,w_+]$ with high probability;
    \item $g^*$ is an arbitrary threshold, but should be chosen to ensure that $G \le g^*$ with high probability;
    \item $\phi_j(s,V) = \sum_{\ell = 0}^j (-1)^{\ell} {j \choose \ell} (1 - \ell s)^V$ is a quantity arising from applying the inclusion-exclusion principle to a union of events in a coupon collector problem \cite[Appendix B]{Joh16}.
\end{itemize}
We set $w_-$, $w_+$, and $g^*$ in the same way as \cite{Joh16}:
\begin{align}
    w_- &= \frac{T}{2} (1-\delta) \label{eq:w-} \\
    w_+ &= \frac{T}{2} (1+\delta) \label{eq:w+} \\
    g^* &= n\Big(\frac{1}{2} + \delta\Big)^L, \label{eq:g*}
\end{align}
for some $\delta > 0$ to be specified later.  The interaction between $\delta$ and $\epsilon$ (see \eqref{eq:T_choice}) turns out to be slightly delicate, and choosing them appropriately is the main difference here compared to \cite{Joh16}. 

The analysis of \cite{Joh16} focuses on the case that \eqref{eq:T_old} holds with equality.  This is without loss of generality, since additional tests can only ever help the DD algorithm.  We similarly assume that \eqref{eq:T_choice} holds with equality.  
In view of the union bound over the $k$ defectives, the goal is to show that $k \Psi_{\nu} \to 0$ for $\nu \in \{1,2,3\}$ in \eqref{eq:three_terms}.  We proceed as follows:
\begin{enumerate}
    \item For $\Psi_1$, it is shown in \cite[Eq.~(39)]{Joh16} that if $L = m(1+\epsilon) \frac{\ln n}{\ln 2}$ (which holds via $L = \frac{T \ln 2}{k}$ and equality holding in \eqref{eq:T_old}) and $k \le c n^m$ for some constant $c$ (with $m$ given in \eqref{eq:m_def}), then 
    \begin{align}
        k\Psi_1 \le c \exp\Big( \frac{L^2}{4w_-} \Big) \exp\bigg( -\Big( \epsilon - \frac{1+\epsilon}{\ln 2} \Big(\delta + \frac{g^* L}{w_-} (1-\delta)\Big)   \Big) m \ln n \bigg). \label{eq:kPsi1}
    \end{align}
    Recall that we are adopting the choice $m=1$; this means that the condition $k \le c n^m$ is trivially satisfied with $c = 1$.  Hence, if we can further establish that $\frac{L^2}{4w_-} = o(1)$ and $\frac{g^* L}{w_-} = o(1)$, it will follow from \eqref{eq:kPsi1} that
    \begin{align}
        k\Psi_1 \le (1+o(1)) \exp\bigg( -\Big( \epsilon - \frac{1+\epsilon}{\ln 2} (\delta + o(1)) \Big) \ln n \bigg).
    \end{align}
    This approaches zero as $n \to \infty$ when $\delta$ is strictly smaller than $\frac{\epsilon \ln 2}{1+\epsilon}$.  For concreteness, we set $\delta = \frac{2}{3}\epsilon$ (note that $\frac{2}{3} < \ln 2$), so that the preceding requirement holds when $\epsilon$ is sufficiently small.
    
    The above-mentioned requirement $\frac{L^2}{4w_-} = o(1)$ follows immediately from the fact that $L = \Theta(\log n)$ and $w_- = \Theta(T) = \Theta(k \log n)$ (with $k = n^{1-o(1)}$).  As for $\frac{g^* L}{w_-}$, the steps in \cite[Eq.~(31)]{Joh16} turn out to be too loose for our purposes, but are easily modified: Combining $L = \frac{T \ln 2}{k}$ with \eqref{eq:w-} gives $\frac{L}{w_-} = \frac{2 \ln 2}{k(1-\delta)}$, and further combining with \eqref{eq:g*} gives
    \begin{align}
        \ln \frac{g^* L}{w_-} = \ln \frac{n}{k} + L\ln\Big( \frac{1}{2} + \delta\Big) + \ln\frac{2 \ln 2}{1-\delta}. \label{eq:log_ratio}
    \end{align}
    The assumption $k = n^{1-o(1)}$ gives $\ln \frac{n}{k} = o(\log n)$, and combining this with $L = \Theta(\log n)$ and $\ln\big( \frac{1}{2} + \delta\big) < 0$ (for small enough $\delta$), it follows that the right-hand side of \eqref{eq:log_ratio} approaches $-\infty$, and hence $\frac{g^* L}{w_-} = o(1)$ as desired.
    \item For $\Psi_2$, we can directly use the following finding from \cite{Joh16} based on McDiarmid's inequality:
    \begin{equation}
        k \Psi_2 \le k\exp\Big( \frac{\delta^2 T}{4 \ln2} (1+o(1)) \Big).
    \end{equation} 
    This approaches zero as $n \to \infty$, since $T = \Theta(k \log n)$.
    \item For $\Psi_3$, we use the following bound \cite[Eq.~(42)]{Joh16} based on Bernstein's inequality, which holds provided that $L \to \infty$ (which we already established) and $\delta \le \frac{1}{4}$:
    \begin{equation}
        k \Psi_3 \le k \exp\bigg( - n \frac{(1/2+\delta)^L}{2/3 + o(1)} \bigg). \label{eq:Psi3_bound}
    \end{equation}
    Recall that $L = \frac{T \ln 2}{k}$; substituting $T$ equaling the right-hand side of \eqref{eq:T_choice} gives $L = (1+\epsilon) \log_2 n$. We proceed by considering the logarithm (base 2) of $n(1/2+\delta)^L$:
    \begin{align}
        \log_2 \Big( n(1/2+\delta)^L \Big) 
        &= \log_2 n + L \log_2 \Big(\frac{1}{2} + \delta\Big) \\
        &=  \big(\log_2 n\big) \bigg[ 1 + (1+\epsilon)\log_2\Big(\frac{1}{2} + \delta\Big) \bigg]. \label{eq:log_cond}
    \end{align}
    Using the above choice $\delta = \frac{2}{3} \epsilon$, a simple Taylor expansion yields the following as $\epsilon \to 0$:\footnote{In fact, a visual plot reveals that $1+(1+\epsilon)\log_2\big(\frac{1}{2} + \frac{2}{3}\epsilon\big)$ is positive for all $\epsilon > 0$.}
    \begin{equation}
        (1+\epsilon)\log_2\Big(\frac{1}{2} + \frac{2}{3}\epsilon\Big) = -1 + \epsilon\Big( \frac{2}{3} \cdot \frac{2}{\ln 2} - 1 \Big) + o(\epsilon).
    \end{equation}
    Hence, since $\frac{2}{3} \cdot \frac{2}{\ln 2} \approx 1.92 > 1$, we have for sufficiently small $\epsilon$ that \eqref{eq:log_cond} is positive and scales as $\Theta(\log n)$, and substituting into \eqref{eq:Psi3_bound} gives $k \Psi_3 \le k \exp\big( - n^{\Theta(1)} \big)\to 0$, as desired.
\end{enumerate}
Since the above analysis holds for arbitrarily small $\epsilon > 0$ (and hence arbitrarily small $\delta > 0$ via $\delta = \frac{2}{3}\epsilon$) when the number of tests satisfies \eqref{eq:T_choice} with equality, the upper bound in Theorem \ref{thm:main1} follows.



{\bf Handling the i.i.d.~prior.} While \cite{Joh16} only considers the combinatorial prior with a fixed value of $k$, the analogous result follows essentially immediately for the i.i.d.~prior, in which $k$ is a random variable.  This is because by a simple concentration argument (e.g., Hoeffding's inequality), as long as $np \to \infty$, it holds that $k = np(1+o(1))$ with probability approaching one.  We can therefore replace the choice $L =  \frac{T \ln 2}{k} (1+o(1))$ by $L =  \frac{T \ln 2}{np} (1+o(1))$, and under the high-probability event $k = np(1+o(1))$, the two are equivalent up to a change in the $o(1)$ term.  Since conditioning on any particular value of $k$ under the i.i.d.~prior brings us back to the combinatorial prior, the desired result follows.